\documentclass[a4paper, reqno]{amsart}
\usepackage{color}
\usepackage{amsmath}
\usepackage{amsthm}
\usepackage{amssymb}
\usepackage{amsmath}
\usepackage{mathcomp}
\usepackage{dsfont}
\usepackage[toc,page]{appendix}
\usepackage{graphicx} 
\usepackage{subfigure}
\usepackage{enumerate}
\usepackage{amsfonts}
\usepackage{amsthm}
\usepackage{amsmath}
\usepackage{times}
\usepackage{amscd}
\usepackage[latin2]{inputenc}
\usepackage{t1enc}
\usepackage{enumitem}
\usepackage[mathscr]{eucal}
\usepackage{indentfirst}
\usepackage{graphicx}
\usepackage{graphics}
\usepackage{pict2e}
\usepackage{epic}
\usepackage{esint}
\usepackage[margin=2.9cm]{geometry}
\usepackage{epstopdf} 
\usepackage{verbatim}
\usepackage{cancel}
\usepackage{amssymb}
\usepackage{xcolor}
\usepackage{dsfont}
\usepackage{hyperref}

\usepackage{physics}

\usepackage[margin=2.9cm]{geometry}

\newtheorem{definition}{Definition}[section]
\newtheorem{lemma}[definition]{Lemma}

\newtheorem{theorem}[definition]{Theorem}

\numberwithin{equation}{section}

\def\bR{\mathbb{R}}
\def\bN{\mathbb{N}}
\def\bP{\mathbb{P}}
\def\bE{\mathbb{E}}

\def\cN{\mathcal{N}}
\def\cU{\mathcal{U}}

\def\cF{\mathcal{F}}
\def\cM{\mathcal{M}}

\def\cG{\mathcal{G}}
\def\cL{\mathcal{L}}
\def\cI{\mathcal{I}}

\def\b1{\mathds{1}}
\def\tr{\mathrm{tr}\;}

\def\cO{\mathcal{O}}

\newcommand{\vertiii}[1]{{\left\vert\kern-0.25ex\left\vert\kern-0.25ex\left\vert #1 
    \right\vert\kern-0.25ex\right\vert\kern-0.25ex\right\vert}}

\begin{document}
\title[]{A note on dependent random variables in quantum dynamics}

\author{Simone Rademacher}
\address{IST Austria, Am Campus 1, 3400 Klosterneuburg, Austria}

\date{\today}

\maketitle

\begin{abstract}
We consider the many-body time evolution of weakly interacting bosons in the mean field regime for initial coherent states.  We show that bounded $k$-particle operators, corresponding to dependent random variables, satisfy both, a law of large numbers and a central limit theorem. 
\end{abstract}

\section{Introduction and main results }
We consider $N$ weakly interacting bosons in the mean-field regime described on $L_s^2( \bR^{3N} )$, the symmetric subspace of $L^2 ( \bR^{3N})$,  by the Hamilton operator 
\begin{align}
\label{eq:ham}
H_N = \sum_{i=1}^N \left( - \Delta_i \right) + \frac{1}{N} \sum_{i<j=1}^N v(x_i-x_j)
\end{align}
with two-body interaction potential $v$ satisfying 
\begin{align}
v^2 \leq C ( 1- \Delta ) 
\end{align}
for a positive constant $C>0$. The mean-field regime is characterized through weak and long-range interactions of the particles.  Trapped Bose gases at extremely low temperatures,  as prepared in the experiments, are known to relax to the ground state.  The ground state $\psi_{N}^{\rm gs}$ of \eqref{eq:ham} exhibits Bose-Einstein condensation \cite{LS_condensation},  i.e.  the associated $\ell$-particle reduced density 
\begin{align}
\label{eq:red-dens-k}
\gamma_{\psi_{N}^{\rm gs}}^{(\ell)} := \tr_{\ell +1, \dots, N} \ket{\psi_{N}^{\rm gs}} \bra{\psi_{N}^{\rm gs}} 
\end{align}
converges in trace norm to 
\begin{align}
\label{eq:BEC2}
\gamma_{\psi_{N}^{\rm gs}}^{(\ell)} \rightarrow \ket{\varphi}  \bra{\varphi}^{\otimes \ell} \quad \text{as} \quad N \rightarrow \infty 
\end{align}
for all $\ell \in \bN$, where $\varphi \in L^2 (\bR^3)$ denotes the condensate wave function, known to be the Hartree minimizer.  However, we remark that the factorized state $\varphi^{\otimes N}$ does not approximate the ground state due to correlations of the particles \cite{GS}. 

\subsection{Law of large numbers} Turning to the probabilistic picture,  the property of Bose-Einstein condensation \eqref{eq:BEC2} implies a law of large numbers for bounded one-particle operators \cite{BKS}.  To be more precise,  for $k\in \bN$ we denote with $O^{(k)}$ a bounded, self-adjoint $k$-particle operator  on $L^2( \bR^{3k})$ and with $\underline{i}_k$ the multi-index 
\begin{align}
\label{eq:multi}
\underline{i}_k = ( i_1, \dots, i_k ) \in \cI_N^{(k)}, \quad \cI_N^{(k)}:= \left\lbrace ( i_1, \dots, i_k ) \in \lbrace 1, \dots, N \rbrace^k \vert \;  i_j \not= i_m \quad \text{for}\quad  j \not= m\right\rbrace .
\end{align} 
Then, we define for fixed $k \leq N$, the $N$-particle operator 
\begin{align}
\label{def:O}
O_{\underline{i}_k}^{(k)} \quad \text{where} \quad \underline{i}_k \in \cI_N^{(k)} 
\end{align}
acting as $O^{(k)}$ on the particles $i_{1}, \dots, i_k $ and as an identity elsewhere.  We consider the operator $O^{(k)}_{\underline{i}_k}$ as a random variable with probability distribution determined through $\psi_N$ by
\begin{align}
\bP_{\psi_N} \left[ O^{(k)}_{\underline{i}_k } \in A \right]  = \bE_{\psi_N} \left[\chi_A \left( O^{(k)}_{\underline{i}_k} \right) \right] =   \expval{\chi_A \left( O^{(k)}_{\underline{i}_k} \right) }{\psi_N}
\end{align}
 where $\chi_A$ denotes the characteristic function of the set $A \subset \bR$.  
 
 For one-particle operators,  factorized states correspond to i.i.d. random variables as for any subsets $A_1,A_2 \subset \bR$ and $i,j \in \cI_N^{(1)}$ with $i \not=j$
 \begin{align}
 \bP_{ \varphi^{\otimes N}}  \left[ O^{(1)}_i \in A_1, \; O^{(1)}_j \in A_2 \right] =& \expval{\chi_{A_1} \left( O^{(1)}_i \right) \chi_{A_2} \left( O^{(1)}_j \right) }{\varphi^{\otimes N}} \notag\\
 =& \expval{\chi_{A_1} \left( O^{(1)} \right) }{\varphi}  \expval{\chi_{A_2} \left( O^{(1)} \right) }{\varphi} \notag\\
 =&  \expval{\chi_{A_1} \left( O^{(1)}_i \right) }{\varphi^{\otimes N}} \expval{\chi_{A_2} \left( O^{(1)}_j \right) }{\varphi^{\otimes N}} \notag\\
 =&  \bP_{ \varphi^{\otimes N}}  \left[ O^{(1)}_i \in A_1 \right]  \;  \bP_{ \varphi^{\otimes N}}  \left[ O^{(1)}_j \in A_2 \right] \; .  \label{eq:independent}
 \end{align}
In particular, for factorized states Chebychef's inequality implies a law of large numbers for the centred averaged sum 
  \begin{align}
\frac{1}{N}O^{(1)}_{N}: = \frac{1}{N} \sum_{i =1}^N \left( O^{(1)}_{i} - \expval{O^{(1)}}{\varphi} \right)    \; . 
 \end{align}
In contrast to one-particle operators,  for $k$-particle operators with $k \geq 2$,  factorized states do not correspond to i.i.d. random variables.  In fact,  for $k\geq 2$,  we have 
 \begin{align}
 \bE_{\varphi^{\otimes N}} \left[ \left( O^{(k)}_{\underline{i}_k} - \expval{O^{(k)}}{\varphi^{\otimes k}} \right)\left( O^{(k)}_{\underline{j}_k} - \expval{O^{(k)}}{\varphi^{\otimes k}} \right)\right] \not= 0  \label{eq:dependence}
 \end{align}
for all $\underline{i}_k \not= \underline{j}_k $ for which  $\underline{i}_k$ contains at least one element of $\underline{j}_k$. We conclude that in this case,  the random variables are correlated and, thus, dependent.  In contrast, whenever $\underline{i}_k$ does not intersect with $\underline{j}_k$, the random variables $O^{(k)}_{\underline{i}_k}$, $O^{(k)}_{\underline{j}_k}$ are independent (following from arguments similarly to \eqref{eq:independent}).  Consequently,  for factorized states,   the random variables $\lbrace O^{(k)}_{\underline{i}_k} \rbrace_{\underline{i}_k \in \cI_N^{(k)} }$  denote a sequence of $m$-dependent of random variables with $m \in \bR$.  Still, as the following theorem shows,  the centred averaged sum 
 \begin{align}
 \label{def:sum-0}
\frac{1}{\tbinom{N}{k}} O^{(k)}_{N}: = \frac{1}{\tbinom{N}{k}} \sum_{\underline{i}_k \in \cI_N^{(k)}} \left( O^{(k)}_{\underline{i}_k} - \expval{O^{(k)}}{\varphi^{\otimes k}} \right)   
 \end{align}
satisfies a law of large numbers. 
 
\begin{theorem}[Law of Large Numbers]
\label{thm:lln}
For $k \in \bN$,  let $O^{(k)}$ denote a self-adjoint bounded $k$-particle operator, $\varphi \in L^2 (\bR^3 )$ and $\psi_{N} \in L^2_s \left( \bR^{3N} \right)$ a bosonic wave function satisfying 
\begin{align}
\label{ass:lln}
\gamma_{\psi_{N}}^{(\ell)} \rightarrow \ket{\varphi}  \bra{\varphi}^{\otimes \ell} \quad \text{as} \quad N \rightarrow \infty \; 
\end{align}
for all $\ell \in \bN$. Then,  for any fixed $k \in \bN$ and $\delta >0$, the averaged sum $O_N^{(k)}$ defined in \eqref{def:sum-0} satisfies 
 \begin{align}
 \bP_{\psi_{N}} \left[ \left\vert \frac{1}{\tbinom{N}{k}}  O^{(k)}_{N} \right\vert > \delta \right] \rightarrow 0 , \quad \text{as} \quad N \rightarrow \infty \; .
 \end{align}
 \end{theorem}

For factorized states, we have $\gamma_{\varphi^{\otimes N}}^{(\ell)} = \ket{\varphi}\bra{\varphi}^{\otimes \ell}$, and a law of large numbers follows from Theorem \ref{thm:lln}.  

In particular, Theorem \ref{thm:lln} shows that the property of condensation \eqref{ass:lln} implies a law of large numbers for bounded $k$-particle operators for fixed $k \in \bN$.  Thus, Theorem \ref{thm:lln} generalizes known results from \cite{BKS} for bounded one-particle operators to $k$-particle operators with fixed $k \in \bN$. 
We recall that the ground state $\psi_{N}^{\rm gs}$ of \eqref{eq:ham} can not be approximated by a factorized state,  nonetheless the condensation property \eqref{eq:BEC2} ensures that bounded $k$-particle operators satisfy a law of large numbers for $\psi_N^{\rm gs}$, too. 

We are interested in the dynamics of initially trapped Bose gases. Removing the trap, the bosons evolve with respect to the Schr\"odinger equation 
\begin{align}
\label{eq:Schroe}
i \partial_t \psi_{N,t} = H_N \psi_{N,t} 
\end{align}
with $H_N$ the mean-field Hamiltonian given in \eqref{eq:ham}.  In the following, we consider coherent initial data, i.e. initial data of the form 
\begin{align}
\label{eq:coherent}
\psi_{N,0} = W ( \sqrt{N} \varphi ) \Omega
\end{align}
where $\Omega$ denotes the vacuum of the bosonic Fock space $\cF = \bigoplus_{n \geq 0} L^2 (\bR^3)^{\otimes_s^n}$ equipped with creation and annihilation operators $a^*(f),a(f)$ for $f\in L^2( \bR^3)$,  $W( f ) = e^{a^*(f) - a (f)} $ denotes the Weyl operator and $\varphi \in H^1( \bR^3)$ the condensate wave function. Coherent states of the form \eqref{eq:coherent} exhibit Bose-Einstein condensation in the quantum state $\varphi$, i.e. they satisfy \eqref{eq:BEC2}.  Thus,  it follows from Theorem \ref{thm:lln} that initially, the random variables $O_{\underline{i}_k}^{(k)}$ satisfy a law of large numbers. 
The property of condensation is preserved along the many-body time evolution \cite[Theorem 3.1]{BPS_book}, i.e.  the $\ell$-particle reduced density $\gamma_{N,t}^{(\ell)}$ associated to $\psi_{N,t}$ satisfies 
\begin{align}
\label{eq:BEC-dyn}
\gamma_{N, t}^{(\ell)} \rightarrow \ket{\varphi_t}  \bra{\varphi_t}^{\otimes \ell} , \quad \text{as} \quad N \rightarrow \infty \quad \text{for all} \quad \ell \in \bN
\end{align}
where $\varphi_t \in H^1( \bR^3)$ denotes the solution to the Hartree equation
\begin{align}
\label{eq:hartree}
i \partial_t \varphi_t = h_{\rm H} (t) \varphi_t, \quad \text{with} \quad  h_{\rm H} (t) = -\Delta + (v* \vert \varphi_t \vert^2) 
\end{align}
with initial data $\varphi_0 = \varphi \in H^1( \bR^3)$ (for further references see e.g. \cite{AGT,AFP,CLS,FKS,GV,KP,RodS,Spohn_mf}).  Theorem \ref{thm:lln} and \eqref{eq:BEC-dyn} show that the centred averaged sum
 \begin{align}
 \label{def:sum}
\frac{1}{\tbinom{N}{k}} O^{(k)}_{N}: = \frac{1}{\tbinom{N}{k}} \sum_{\underline{i}_k \in \cI_N^{(k)}} \left( O^{(k)}_{\underline{i}_k} - \expval{O^{(k)}}{\varphi_t^{\otimes k}} \right)   
 \end{align}
satisfies a law of large numbers for positive times $t>0$, too, i.e.  for any $\delta >0$
  \begin{align}
 \bP_{\psi_{{N,t}}} \left[ \left\vert \frac{1}{\tbinom{N}{k}}  O^{(k)}_{N,t} \right\vert > \delta \right] \rightarrow 0 , \quad \text{as} \quad N \rightarrow \infty \; \quad \text{for all } \quad t \in \bR \; .
 \end{align}

\subsection{Central limit theorem} While the law of large numbers characterizes the mean of the probability distribution, fluctuations around are governed through the central limit theorem.  Before stating our result on a central limit theorem for fluctuations of order $O(N^{k-1/2})$, we introduce some notations. For a bounded $k$-operator $O^{(k)}$ and $\varphi \in L^2( \bR^3)$, we define 
\begin{align}
\label{eq:def-h}
 & \left( \overline{\varphi}_t^{\otimes (k-1)} O^{(k)} \varphi_t^{\otimes k} \right)_j (x)\notag\\
&\hspace{0.5cm} :=  \int dx_1 \dots dx_{i-1} dx_{i+1} \dots  dx_k dy_1 \dots dy_k \;  \widetilde{O}^{(k)}(x_1, \dots,x_{i-1},x,x_{i+1}, \dots, x_k; y_1, \dots y_k )\notag\\
& \hspace{3cm} \times  \prod_{\substack{i=1\\ i \not= j}}^k \;  \overline{\varphi}_t (x_i)   \prod_{m=1}^k \varphi_t (y_m)
\end{align}
and furthermore, for $t\in\bR$, $0 \leq s \leq t$ and $j \in \lbrace 1, \dots, k \rbrace$ the function $f_{s;t}^{(j)}$ by 
\begin{align}
\label{def:f}
i \partial_s f_{s;t}^{(j)} = \left( h_{\rm H} (s) + K_{1,s} - K_{2,s} J \right) f_{s;t}^{(j)}, \quad \text{with} \quad f_{t;t}^{(j)}= q_t \left( \overline{\varphi}_t^{\otimes (k-1)} O^{(k)} \varphi_t^{\otimes k} \right)_j
\end{align}
with the anti-linear operator $J f = \overline{f}$ for any $f \in L^2( \bR^3)$, $q_t = 1 - \ket{\varphi_t}\bra{\varphi_t}$,  the Hartree Hamiltonian $h_{\rm H}$ defined in \eqref{eq:hartree} and the operators 
\begin{align}
\label{eq:def-K}
K_{1,t} (x;y) = \varphi_t(x) v(x-y) \varphi_t(y),\quad K_{2,t} (x;y) = \varphi_t (x) v(x-y) \varphi_t (y) \; . 
\end{align}

\begin{theorem}[Central Limit Theorem]
\label{thm:CLT}
For $k,N \in \bN$ with $k \leq N$,  let $O^{(k)}$ be a self-adjoint bounded $k$-particle operator and  $\varphi_t$ the solution to the Hartree equation \eqref{eq:hartree} with initial datum $\varphi_0 = \varphi \in H^1( \bR^3)$. 
Let  $\psi_{N,t} \in L^2_s( \bR^{3N})$ denote the solution to the Schr\"odinger equation \eqref{eq:Schroe} with initial datum of the form $\psi_{N,0} = W( \sqrt{N} \varphi ) \Omega$. 

Let $a,b \in \bR$ with $a<b$,  then there exists a constant $C_{a,b,k} >0$ such that the centred averaged sum $O_{N,t}^{(k)}$ defined in \eqref{def:sum} satisfies 
\begin{align}
\label{eq:CLT-thm}
\left\vert \bP_{\psi_{N,t}} \left[ N^{-k+1/2} O_{N,t}^{(k)} \in [a, b ] \right] - \bP \left[ \cG_t \in [a, b ] \right] \right\vert \leq C_{a,b,k} \;  e^{C \vert t \vert}N^{-1/12}  
\end{align}
where $\cG_t$ denotes the Gaussian random variable with variance given by 
 \begin{align}
\label{eq:cov}
\sigma^2_t = \sum_{i,j=1}^k \cM_{t;0} (i,j) = \sum_{i,j=1}^k \bra{f_{0;t}^{(i)}}\ket{f_{0;t}^{(j)}}
\end{align}
\end{theorem}

We remark that for a factorized state, we can explicitly compute the variance 
\begin{align}
\sigma_N^2 =& \bE_{\varphi^{\otimes N}} \left[ \left( O^{(k)}_N\right)^2 \right]  - \bE_{\varphi^{\otimes N}} \left[ O^{(k)}_N \right]^2 \notag\\
=&  \sum_{\underline{i}_k, \underline{j}_k \in I_N^{(k)}} \bE_{\varphi^{\otimes N}} \left[ \widetilde{O}_{\underline{i}_k}^{(k)}\widetilde{O}_{\underline{j}_k}^{(k)} \right] - \left(  \sum_{\underline{i}_k\in I_N^{(k)}} \bE_{\varphi^{\otimes N}} \left[ \widetilde{O}_{\underline{i}_k}^{(k)} \right] \right)^2 \label{eq:variance}
\end{align}
where we introduced the centred $k$-particle operator 
\begin{align}
\label{eq:def-Otilde}
\widetilde{O}^{(k)} = O^{(k)} - \expval{O^{(k)}}{\varphi^{\otimes k}} \; .
\end{align}
The last sum of the r.h.s. of \eqref{eq:variance} vanishes.  Furthermore, the first sum vanishes whenever $\underline{j}_k$ does not intersect with $\underline{i}_k$ and we find for the remaining terms 
\begin{align}
\sigma_N^2 =& \sum_{i,j=1}^k \frac{N \cdots (N-2k+1)}{k!(k-1)!} \;  \cM_{\varphi^{\otimes N}} (i,j) + O \left( N^{2k-2} \right)
\end{align}
using the definition  
\begin{align}
\label{eq:cov-prod}
\cM_{\varphi^{\otimes N}} (i,j)=& \braket{ \left( \overline{\varphi}^{\otimes (k-1)} \widetilde{O}^{(k)} \varphi^{\otimes k} \right)_i}{ \left( \overline{\varphi}^{\otimes( k-1)} \widetilde{O}^{(k)} \varphi^{\otimes k} \; \right)_j}_{L^2(\bR^3)} \notag\\
=& \braket{ q \left( \overline{\varphi}^{\otimes (k-1)} O^{(k)} \varphi^{\otimes k} \right)_i}{q  \left( \overline{\varphi}^{\otimes( k-1)} O^{(k)} \varphi^{\otimes k} \; \right)_j}_{L^2(\bR^3)} 
\end{align}
with $q = 1 - \ket{\varphi}\bra{\varphi}$ and \eqref{eq:def-h}. In particular, we observe that the variance scales as $\sigma_N^2 = O(N^{2k-1}) $ and thus, we expect fluctuations to be $O(N^{k-1/2})$. 

We observe that Theorem \ref{thm:CLT} shows that the fluctuations of the many-body dynamics scale similarly to the fluctuations of a factorized state. Moreover, for  $t=0$ the variance $\sigma_0^2$ of the many-body dynamics defined in \eqref{eq:cov} agrees with the covariance matrix $\cM_{\varphi^{\otimes N}} (i,j) $ in \eqref{eq:cov-prod} of a factorized state.  

We remark that for $k=1$, i.e. considering bounded one-particle observables, Theorem \ref{thm:CLT} generalizes known results \cite{BKS,BSS} to more general one-particle observables. This generalization is due a different strategy of the proof of Theorem \ref{thm:char} than in \cite{BKS,BSS} . We follow the ideas of \cite{BSS},  however,  we use as a first step in Lemma \ref{lemma:step1} directly the norm approximation \eqref{eq:norm-approx} of the many-body time evolution (for more details see Section \ref{sec:proof}). 

Recently, for one-particle operators the probability distribution's tails were characterized through large deviation estimates \cite{KRS,RSe}, showing that 
\begin{align}
\lim_{N \rightarrow \infty}\frac{1}{N} \log \bP_{\psi_{N,t}} \left[ \tfrac{1}{N} O_{N,t}^{(1)} > x \right]=  - \frac{x^2}{\| \widetilde{f}_{0;t}^{(1)}\|_2^2} + O(x^{5/2} )
\end{align}
for sufficiently small $x \leq C e^{-e^{C \vert t \vert }}$ where $\widetilde{f}_{t,0}^{(1)}$ is defined similarly to \eqref{def:f}, but using the projected kernels $\widetilde{K}_{j,s} (x,y)= q_s K_{j,s} (x,y) q_s$. 

Furthermore, for one-particle operators, a central limit theorem is proven for stronger particles' interactions in the intermediate regime \cite{R}, interpolating between the mean-field and the Gross-Pitaevski regime.  In the Gross-Pitaevski regime of singular particles' interaction, a central limit theorem is proven for quantum fluctuations in the ground state \cite{RS_CLT}, too. 

Theorem \ref{thm:CLT} follows from an approximation of the random variable's characteristic function given in the following: 

\begin{theorem}\label{thm:char}
Under the same assumptions as in Theorem \ref{thm:CLT}, we have 
\begin{align}
&\left\vert \bE_{\psi_{N,t}} \left[ e^{i N^{-k+1/2} O_{N,t}^{(k)}} \right]  - e^{- \sigma_t^2 /2}  \right\vert\leq C_k e^{C \vert t \vert } \| O^{(k)}\|_{\rm op} \sum_{\ell=1}^{2k-1} N^{-j/2} \big( 1 + \sum_{i,j} \| O^{(k)} \|_{\rm op}^{2}  \big)^{(j+1)/2} 
\end{align}
\end{theorem}
In the following, we will now first turn to the proof of Theorem \ref{thm:lln} in Section \ref{sec:lln}, then prove Theorem \ref{thm:CLT} from Theorem \ref{thm:char} in Section \ref{sec:CLT} and finally prove Theorem \ref{thm:char} in Section \ref{sec:char}.

\section{Proof of Theorem \ref{thm:lln}} 
\label{sec:lln}

We  generalize ideas from \cite{BKS} on a law of large numbers for bounded one-particle observables to the case of $k$-particle operators. 

 \begin{proof}
By Chebycheff's inequality, we have 
  \begin{align}
  \label{eq:LLN1}
 \bP_{\psi_{N}} \left[ \left\vert \tfrac{1}{\binom{N}{k}}   O^{(k)}_N \right\vert > \delta \right] \leq & \frac{1}{ \binom{N}{k}^2 \delta^2} \bE_{\psi_N} \left[ \big\vert \sum_{\underline{i}_{k} \in \cI_N^{(k)}} \widetilde{O}^{(k)}_{\underline{i}_k} \big\vert^2  \right] 
\end{align}
where we used the notation $\widetilde{O}^{(k)}$ defined in \eqref{eq:def-Otilde}.  Furthermore, we denote with $\sharp \lbrace \underline{i}_k, \underline{j}_k\rbrace$ the number of elements of $\underline{i}_k$ agreeing with $\underline{j}_k$. Then, we can write 
\begin{align}
\big\vert \sum_{\underline{i}_k \in \cI_N^{(k)}} \widetilde{O}^{(k)}_{\underline{i}_k } \big\vert^2  =\sum_{\ell =0}^k \;\;\; \sum_{ \substack{\underline{i}_k, \underline{j}_k \in \cI_N^{(k)} \\ \sharp \lbrace \underline{i}_k, \underline{j}_k\rbrace = \ell } }  \widetilde{O}^{(k)}_{\underline{i}_k} \widetilde{O}^{(k)}_{\underline{j}_k} 
\end{align}
we can express the r.h.s. of \eqref{eq:LLN1} in terms of $j$-particle reduced density matrices defined in \eqref{eq:red-dens-k} and find 
\begin{align}
\bE_{\psi_N} \left[ \big\vert \sum_{i_1, \dots , i_k \in I_N} \widetilde{O}^{(k)}_{\lbrace i_1, \dots , i_k \rbrace } \big\vert^2  \right] =&   \sum_{\ell=0}^k   \;\;\; \sum_{ \substack{\underline{i}_k, \underline{j}_k \in \cI_N^{(k)} \\ \sharp \lbrace \underline{i}_k, \underline{j}_k\rbrace = \ell } } \tr \gamma_{\psi_N}^{(2k-\ell)}  \widetilde{O}^{(k)}_{\underline{i}_k} \widetilde{O}^{(k)}_{\underline{j}_k} \label{eq:LLN2}
\end{align}

Plugging \eqref{eq:LLN2} into the r.h.s. of \eqref{eq:LLN1}, we find 
 \begin{align}
 \bP_{\psi_{N}} \left[ \left\vert   O^{(k)}_N \right\vert > \delta \right] \leq & \frac{1}{\tbinom{N}{k}^2 \delta^2}\sum_{\ell=0}^k \;\;\;  \sum_{ \substack{\underline{i}_k, \underline{j}_k \in \cI_N^{(k)} \\ \sharp \lbrace \underline{i}_k, \underline{j}_k\rbrace = \ell } } \tr \gamma_{\psi_N}^{(2k-\ell)} \left( \widetilde{O}^{(k)}_{\underline{i}_k} \widetilde{O}^{(k)}_{\underline{j}_k} \right) \label{eq:LLN3}
\end{align}
For $\ell =0$, the term of the sum of the r.h.s. of \eqref{eq:LLN3} is given by 
\begin{align}
 \frac{1}{\tbinom{N}{k}^2 \delta^2}   \sum_{ \substack{\underline{i}_k, \underline{j}_k \in \cI_N^{(k)} \\ \sharp \lbrace \underline{i}_k, \underline{j}_k\rbrace = 0 } } \tr \gamma_{\psi_N}^{(2k}  \widetilde{O}^{(k)}_{\underline{i}_k} \widetilde{O}^{(k)}_{\underline{j}_k}   =   \frac{\tbinom{N}{2k}}{\tbinom{N}{k}^2 \delta^2}   \tr \gamma_{\psi_N}^{(2k)}  \left( \widetilde{O}^{(k)} \otimes \widetilde{O}^{(k)} \right)
 \leq C_k  \tr \gamma_{\psi_N}^{(2k)} \left( \widetilde{O}^{(k)} \otimes \widetilde{O}^{(k)}\right) . 
\end{align}
Since $\psi_{N}$ exhibits Bose-Einstein condensation, it follows by assumption \eqref{ass:lln} 
\begin{align}
   \tr \gamma_{\psi_N}^{(2k)} \left(  \widetilde{O}^{(k)} \otimes \widetilde{O}^{(k)} \right)^{(2k)} \rightarrow    \tr  \ket{\varphi_t} \ket{\varphi_t}^{\otimes (2k)} \left(  \widetilde{O}^{(k)} \otimes \widetilde{O}^{(k)} \right)^{(2k)} \quad \text{as} \quad N \rightarrow  \infty \label{eq:LLN4}
\end{align}
and by definition \eqref{eq:def-Otilde} of $\widetilde{O}^{(k)}$, we arrive at
\begin{align}
  \tr  \ket{\varphi} \ket{\varphi}^{\otimes (2k)} \left(  \widetilde{O}^{(k)} \otimes \widetilde{O}^{(k)} \right)^{(2k)} =0 .\label{eq:LLN5}
\end{align}
For $\ell \geq 1$, the terms of the sum of the r.h.s. of \eqref{eq:LLN3} consists of $(2k-\ell)$-particle operators whose expectation values are computed with $(2k-\ell)$-particle operators.  In particular, we find 
\begin{align}
\frac{1}{\tbinom{N}{k}^2 \delta^2}\sum_{\ell=1}^k  \;\;\; \sum_{ \substack{\underline{i}_k, \underline{j}_k \in \cI_N^{(k)} \\ \sharp \lbrace \underline{i}_k, \underline{j}_k\rbrace = \ell } } \tr \gamma_{\psi_N}^{(2k-\ell)} \left( \widetilde{O}^{(k)}_{\underline{i}_k} \widetilde{O}^{(k)}_{\underline{j}_k} \right) \leq C_k \| O^{(k)}\|_{\rm op}^2 \sum_{\ell=1}^k  \frac{\binom{N}{2k-\ell}}{\binom{N}{k}^2} \leq C_k \| O^{(k)}\|_{\rm op}^2 \sum_{\ell=1}^k \frac{1}{N^{\ell}}\rightarrow 0  \label{eq:LLN6}
\end{align}
as $N\rightarrow \infty$. 

We conclude with \eqref{eq:LLN5}, \eqref{eq:LLN6} and \eqref{eq:LLN3} by 
\begin{align}
 \bP_{\psi_{N}} \left[ \left\vert  \tfrac{1}{\binom{N}{k}} O^{(k)}_N \right\vert > \delta \right] \rightarrow 0 \quad \text{as} \quad N \rightarrow \infty .
\end{align}
\end{proof}

\section{Proof of Theorem \ref{thm:CLT}}
\label{sec:CLT}

We use standard arguments from probability theory to prove Theorem \ref{thm:CLT} from Theorem \ref{thm:char}.  We follow the arguments from \cite[Corollary 1.2 ]{BSS}. 

\begin{proof}
We consider the difference 
\begin{align}
\bP_{\psi_{N,t}} & \left[ N^{-k+1/2} O_{N,t}^{(k)} \in \left[a, b\right] \right] - \bP \left[ \cG_{t} \in \left[ a,b \right] \right]  \notag\\
&= \expval{\chi_{\left[ a, b \right]} \left( N^{-k+1/2} O_{N,t}^{(k)}\right)}{\psi_{N,t}} - \frac{1}{\sqrt{2\pi} \sigma_t}\int dx \;   e^{-\frac{x^2}{2 \sigma_t^2} }   \chi_{\left[ a,b \right]  } (x) \notag\\
&= \bE_{\psi_{N,t}} \left[ \chi_{[a,b]}\left(  N^{-k+1/2}  O_{N,t}^{(k)} \right) \right] - \bE   \left[  \chi_{[a,b]} \left( \cG_t \right)  \right] \label{eq:CLT1}
\end{align}
where $\chi_{[a,b]}$ denotes the characteristic function of the set $[ a,b]$.   We observe that for $g \in  L^1 ( \bR)$ with Fourier transform $\widehat{g} \in L^1\left( \bR, (1+  s^{2k}) \; ds \right)$,  we have on the one hand 
\begin{align}
\bE_{\psi_{N,t}} \left[ g\left(  N^{-k+1/2}  O_{N,t}^{(k)} \right) \right] =& \expval{g\left( N^{-k+1/2}O_{N,t}^{(k)} \right) }{\psi_{N,t}} \notag\\
 =&  \int d\tau \;  \widehat{g}(\tau) \expval{e^{i \tau  N^{-k+1/2} O_{N,t}^{(k)}}}{\psi_{N,t}} 
\end{align}
and on the other hand 
\begin{align}
\bE  \left[  g \left( \cG_t \right)   \right]  = \frac{1}{\sqrt{2\pi \sigma^2}} \int dx \;  g (x) \;  e^{- \frac{x^2}{2 \sigma_t^2}} =  \int d\tau \;  \widehat{g} (\tau) \; e^{- \frac{\tau^2 \sigma_t^2}{2}} \; . 
\end{align}
and,  in particular by Theorem \ref{thm:char} 
\begin{align}
\big\vert\bE_{\psi_{N,t}} &  \left[ g\left(  N^{-k+1/2}  O_{N,t}^{(k)} \right) \right] -  \bE  \left[  g \left( \cG_t \right) \right]  \big\vert \notag\\
& \leq C_k e^{C \vert t \vert } \| O^{(k)} \|_{\rm op}  \int d\tau \; \vert \widehat{g} (\tau \vert \; \sum_{\ell =1}^{2k - 1} N^{-j/2} \left( 1 + \tau^2 \| O^{(k)} \|_{\rm op}^2 \right)^{(j+1)/2} .  \label{eq:CLT2}
\end{align}
Thus, in order to find an estimate for \eqref{eq:CLT1}, we shall find an approximation from above $f_{+, \varepsilon}$ and from below $f_{-, \varepsilon}$ of the characteristic function $\chi_{[a,b]}$ which satisfy $f_{-,\varepsilon}, f_{+,\varepsilon} \in L^1( \bR^3)$ and $\widehat{f}_{-,\varepsilon}, \widehat{f}_{+,\varepsilon} \in L^1( \bR, (1+s^{2k})ds)$.  For this, let $\eta \in C_0^\infty ( \bR )$ with $\eta \geq 0$, $\eta (s) =0 $ for all $\vert s \vert \geq 1$ and $\int ds \; \eta (s) =1$. Furthermore, for $\varepsilon >0$, let $\eta_{\varepsilon} (s) = \varepsilon^{-1} \eta (s/ \varepsilon )$. Then, for any $\varepsilon >0$, we define 
\begin{align}
f_{-,\varepsilon} := \chi_{[a+ \varepsilon,b-\varepsilon]} * \eta_\varepsilon, \quad \text{resp.} \quad f_{+,\varepsilon} := \chi_{[a- \varepsilon,b+\varepsilon]} * \eta_\varepsilon
\end{align}
which satisfy 
\begin{align}
\label{eq:CLT3}
f_{-,\varepsilon} \leq \chi_{[a,b]} \leq f_{+,\varepsilon} \; .
\end{align}
Moreover,  the Fourier transform is given by  
\begin{align}
\label{eq:CLT4}
\widehat{f}_{-,\varepsilon} (\tau)  = -i \tau^{-1} \left( e^{i\tau (b- \varepsilon)} - e^{i\tau (a + \varepsilon)} \right)  \widehat{\eta} ( \varepsilon \tau ) 
\end{align}
Thus it follows from \eqref{eq:CLT1}, \eqref{eq:CLT3} 
\begin{align}
 \bP_{\psi_{N,t}} & \left[ N^{-k+1/2} O_{N,t}^{(k)} \in [a,b]\right] - \bP \left[ \cG_t \in [a,b]\right] \notag\\
 &\geq  - \big\vert \bE \left[  f_{-,\varepsilon} \left(  \cG_t \right) \right] - \bE_{\psi_{N,t}}\left[  f_{-,\varepsilon} \left(  N^{-k+1/2} O_N^{(k)} \right) \right]\big\vert  - \big\vert  \bE \left[  f_{-,\varepsilon} \left(  \cG_t \right) \right] - \bE \left[ \chi_{[a, b]} \left(\cG_t \right)\right]\big\vert 
\end{align}
and with \eqref{eq:CLT2}, \eqref{eq:CLT4} we arrive at 
\begin{align}
 \bP_{\psi_{N,t}} & \left[ N^{-k+1/2} O_{N,t}^{(k)} \in [a,b]\right] - \bP \left[ \cG_t \in [a,b]\right] \notag\\
 &\geq - C_k e^{C \vert t \vert }\sum_{\ell =1}^{2k-1} N^{-j}\left( \vert a - b \vert \varepsilon^{-1} + \varepsilon^{-2} \right)^{(j+1)/2)} - C \varepsilon  \;. 
\end{align}
Similarly, using $f_{+,\varepsilon}$ we have 
\begin{align}
 \bP_{\psi_{N,t}} & \left[ N^{-k+1/2} O_{N,t}^{(k)} \in [a,b]\right] - \bP \left[ \cG_t \in [a,b]\right] \notag\\
 &\leq C_k e^{C \vert t \vert }\sum_{\ell =1}^{2k-1} N^{-j}\left( \vert a - b \vert \varepsilon^{-1} + \varepsilon^{-2} \right)^{(j+1)/2)} + C \varepsilon 
\end{align}
Now, we optimize with respect to $\varepsilon>0$ and arrive at \eqref{eq:CLT-thm}. 
\end{proof}

\section{Proof of Theorem \ref{thm:char}}
\label{sec:char}

\subsection{Fluctuations around the Hartree dynamics. }  In the following, we consider the bosonic $N$-particle wave function $\psi_{N,t}$ as an element of the bosonic Fock space $\cF = \bigoplus_{n \geq 0} L^2 (\bR^3)^{\otimes_s^n}$ with creation and annihilation operators $a^*(f),a(f)$ for $f \in L^2( \bR^3)$.  Theorem \ref{thm:char} characterizes the fluctuations around the Hartree dynamics which are well described by the approximation of the many-body time evolution \cite[Theorem 4.1]{BPS_book} resp. \cite[Proposition 3.3]{BSS} in $L^2( \bR^{3N})$-norm
\begin{align}
\label{eq:norm-approx}
\|  \psi_{N,t} - W ( \sqrt{N} \varphi_t ) \cU_{\infty} (t;0) \Omega \| \leq C \vert t \vert \; N^{-1/2} \; 
\end{align}
where the limiting dynamics $\cU_{\infty} (t;0)$ is given by 
\begin{align}
\label{eq:Uinfty}
i\partial_t \cU_\infty (t;0) = \cL (t) \; \cU_\infty (t;0), \quad \cU_\infty(0;0) = 1
\end{align}
with generator 
\begin{align}
\cL(t)= d\Gamma \left( h_{\rm H} (t) + K_{1,t} \right) + \int dxdy \left( K_{2,t} (x;y) a_x^*a_y^* + \overline{K}_{2,t} (x;y) a_x a_y \right) 
\end{align}
Here,  $d\Gamma (A) = \int dxdy \; A(x;y) a_x^* a_y$ denotes the second quantization  of an operator $A$  on $L^2( \bR^3)$,  $h_H (t)$ the Hartree Hamiltonian defined in \eqref{eq:hartree},  and $K_{j,t}$ denote the operators defined in \eqref{eq:def-K}.  For further references, see also \cite{ChenX,Hepp,GM,LNS}. The generator $\cL_\infty(t)$ is quadratic in creation and annihilation operators, and thus \cite[Theorem 2.2]{BKS} (see also \cite{Lea,R} gives rise to a Bogoliubov dynamics, i.e. there exists bounded operators $U(t;0), V(t;0)$ on $L^2( \bR^3)$ such that for $f,g \in L^2( \bR^3)$ and the operator $A(f,g) = a(f) + a^*(\overline{g})$, we have 
\begin{align}
\label{eq:bogo}
\cU_\infty^* (t;0)A(f,g) \cU_\infty (t;0) = A( \Theta (t;0) (f,g)),\quad \text{with} \quad \Theta(t;0) = \begin{pmatrix}
U(t;0) &  J V(t;0) J \\
V(t;0) & J U(t;0) J  \\
\end{pmatrix}
\end{align}
where $Jf = \overline{f}$ for any $f \in L^2( \bR^3)$. In particular,  for the operator 
\begin{align}
\label{eq:phi}
\phi(f) = a(f) + a^*(f), \quad \text{for} \quad f \in L^2 (\bR^3)
\end{align}
we have from \eqref{eq:bogo} 
\begin{align}
\label{eq:bogo2}
\cU_\infty^* (t;0)\phi(f) \cU_\infty (t;0) = \phi (  (U(t;0)  + JV(t;0) )f) \; 
\end{align}
and it follows from \cite[ Theorem 2.2 and subsequent Remark]{BKS},  that  
\begin{align}
\label{eq:bogodyn}
i \partial_s \left( U(t;s) + JV(t;s) \right) f=\left( h_H(s) + K_{1,s} - K_{2,s} J \right) \left( U(t;s) + JV(t;s) \right) f \; . 
 \end{align}
Comparing with \eqref{def:f}, we notice that the variance $\sigma_t$ defined in \eqref{eq:cov} is determined by the limiting Bogoliubov dynamics \eqref{eq:Uinfty}, i.e. the fluctuations' quasi-free approximation. 

\subsection{Proof of Theorem \ref{thm:char} }  \label{sec:proof} The proof \ref{thm:char} is split into three steps covered by Lemma \ref{lemma:step1} to Lemma \ref{lemma:step3}. 

For the first step, Lemma \ref{lemma:step1},  we use similarly to the strategy in \cite{R,RS_CLT}, directly the norm approximation \eqref{eq:norm-approx}. This allows to consider more general $k$- (resp. one-) particle operators than in \cite{BKS,BSS} where the difference of the limiting fluctuation dynamics $\cU_\infty (t;0)$ defined in \eqref{eq:Uinfty} to the full many-body dynamics was estimated in \eqref{eq:step11} by Duhamel's formula and a Gronwall estimate. The remaining steps use similar ideas as in \cite{BKS,BSS}. 

\begin{lemma}
\label{lemma:step1}
Under the same assumptions as in Theorem \ref{thm:CLT},  let 
\begin{align}
\xi_{N,t} = W(\sqrt{N} \varphi_t) \cU_\infty (t;0) \Omega \label{eq:def-xi}
\end{align}
Then, there exists $C>0$ such that 
\begin{align}
\Big\vert  \bE_{\psi_{N,t}} & \left[ e^{ iN^{-k+1/2} O^{(k)}_N} \right] -  \bE_{\xi_{N,t}} \left[ e^{ iN^{-k+1/2} O^{(k)}_{N,t}} \right] \Big\vert \leq C \vert t \vert N^{-1/2}  \; . 
\end{align} 
\end{lemma}

\begin{proof} We have 
\begin{align}
 \bE_{\psi_{N,t}} & \left[ e^{ iN^{-k+1/2} O^{(k)}_{N,t}} \right] -  \bE_{\xi_{N,t}} \left[ e^{ iN^{-k+1/2} O_{N,t}^{(k)}} \right] \notag\\
=&    \vert  \bra{\psi_{N,t}}  e^{i \tau  N^{-k+1/2}  \widetilde{O}^{(k)}_{N,t}}  \ket{\psi_{N,t} - \xi_{N,t} } \vert + \vert  \bra{\psi_{N,t} - \xi_{N,t}}  e^{i \tau  N^{-k+1/2}   \widetilde{O}_{N,t}^{(k)}) } \ket{ \xi_{N,t} } \vert  \; .\label{eq:step11}
\end{align}
The operator $O^{(k)}$ is a self-adjoint operator,  thus $\| e^{i N^{-k+1/2} O^{k}_N} \|_{\rm op} \leq 1 $ and we find with \eqref{eq:def-xi} and  \eqref{eq:norm-approx} 
\begin{align}
\Big\vert  \bE_{\psi_{N,t}} & \left[ e^{ iN^{-k+1/2} O^{(k)}_{N,t}} \right] -  \bE_{\xi_{N,t}} \left[ e^{ iN^{-k+1/2} O^{(k)}_{N,t}} \right] \Big\vert \leq C \vert t \vert N^{-1/2} 
\end{align}
\end{proof} 

\begin{lemma}
\label{lemma:step2}
Under the same assumptions as in Theorem \ref{thm:CLT},  let $\phi (f)$ be defined as in \eqref{eq:phi} and $h_t = \sum_{j=1}^k h_{j,t} \in L^2( \bR^3)$ defined with \eqref{eq:def-h} by 
\begin{align}
\label{eq:defh2}
h_{j,t} = \left( \overline{\varphi}_t^{\otimes (k-1)} O^{(k)} \varphi_t^{\otimes k}\right)_j
\end{align}
Then, there exists $C>0$ such that 
\begin{align}
\left\vert \bE_{\xi_{N,t}}  \left[ e^{ i N^{-k+1/2} O_{N,t}^{(k)}} \right] - \bE_{\cU_\infty (t;0) \Omega} \left[ e^{i \phi (h_t)} \right] \right\vert  \leq C_k \; e^{C \vert t \vert }  \| O^{(k)}\|_{\rm op} \sum_{j=1}^{2k-1} N^{-j/2} \left( 1 + \| O^{(k)} \|_{\rm op}^2 \right)^{(j+1)/2} \; . 
\end{align}
\end{lemma}

\begin{proof} We observe that on the bosonic Fock space, we have 
\begin{align}
\bE_{\xi_{N,t}} \left[ e^{ i N^{-k+1/2} O_{N,t}^{(k)}} \right] =  \expval{e^{i  N^{-k+1/2} d\Gamma^{(k)} ( \widetilde{O}^{(k)}) }}{\xi_{N,t}}
\end{align}
where we used the second quantization 
\begin{align}
d\Gamma^{(k)} (\widetilde{O}^{(k)}) = \int dx_1 \dots dx_k dy_1 \dots dy_k \; \widetilde{O}^{(k)}(x_1, \dots, x_k; y_1, \dots y_k ) a_{x_1}^* \dots a_{x_k}^*  a_{y_1} \dots a_{y_k}   
\end{align}
of the operator $\widetilde{O}^{(k)}$ defined in \eqref{eq:def-Otilde}
to express the sum $O^{(k)}_{N,t}$ on the Fock space. Thus,  recalling \eqref{eq:def-xi}, we shall estimate the expectation value 
\begin{align}
\bE_{\xi_{N,t}} \left[  e^{ N^{-k+1/2}    O^{(k)}_{N,t} } ) \right]
=& \expval{W^* ( \sqrt{N} \varphi_t )  e^{i \tau  N^{-k+1/2} d \Gamma^{(k)} (\widetilde{O}^{(k)} ) } W( \sqrt{N} \varphi_t ) } {\cU_\infty (t;0) \Omega} \notag\\
=& \bE_{\cU_\infty (t;0) \Omega} \left[W^*(\sqrt{N} \varphi_t ) e^{i \tau N^{-k + 1/2} d \Gamma^{(k)} ( \widetilde{O}^{(k)})}  W(\sqrt{N} \varphi_t ) \right]  \label{eq:CLT2-1}
\end{align}
In order to compute the operator 
\begin{align}
\label{eq:def-ON}
\cO_{N,t} =  N^{-k+1/2} W^* ( \sqrt{N} \varphi_t ) \; d\Gamma^{(k)} ( \widetilde{O}^{(k)} )  W( \sqrt{N} \varphi_t ) 
\end{align}
we use the Weyl operators' shifting properties on creation and annihilation operators, i.e. 
\begin{align}
W^* ( \sqrt{N} \varphi_t ) a_x  W( \sqrt{N} \varphi_t ) = a_x + \sqrt{N} \varphi_t(x), \quad W^* ( \sqrt{N} \varphi_t ) a^*_x  W( \sqrt{N} \varphi_t ) = a_x^* + \sqrt{N}\;  \overline{\varphi}_t(x) \; .
\end{align}
We find 
\begin{align}
\cO_{N,t}
=& N^{-k + 1/2} \int dx_1 \dots dx_k dy_1 \dots dy_k \; \widetilde{O}^{(k)}(x_1, \dots, x_k; y_1, \dots y_k ) \notag\\
& \hspace{0.5cm} \times \left( a_{x_1}^* + \sqrt{N} \; \overline{\varphi}_t(x_1) \right) \dots \left( a_{x_k}^* + \sqrt{N} \; \overline{\varphi}_t(x_k) \right)   \left( a_{y_1} + \sqrt{N}\varphi_t (y_1)  \right) \dots  \left( a_{y_k} + \sqrt{N}\varphi_t (y_k)  \right) 
 \label{eq:CLT2-id1}
\end{align}

We observe that the leading order term $O(N^{k})$  vanishes by definition of $\widetilde{O}^{(k)} $ in \eqref{eq:def-Otilde}.  Thus, the first non-vanishing leading order term is $O(N^{k-1/2})$, and, in particular,  we have with \eqref{eq:phi} and \eqref{eq:defh2}
\begin{align}
\cO_{N,t} =& \int dx_1 \dots dx_k dy_1 \dots dy_k \;  \widetilde{O}^{(k)}(x_1, \dots, x_k; y_1, \dots y_k ) \notag\\
& \hspace{2cm} \times \left( \prod_{m=1}^k \varphi_t (y_m) \sum_{j=1}^k a^*_{x_j} \prod_{\substack{i=1\\ i \not= j}}^k \;  \overline{\varphi}_t (x_i) +  \prod_{m=1}^k \overline{\varphi}_t (x_m) \sum_{j=1}^k a_{y_j} \prod_{\substack{i=1\\ i \not= j}}^k   \varphi_t (y_i) \right) + \mathcal{R}_N  \notag\\
=&  \phi ( h_t) + \mathcal{R}_N  \label{eq:leading-term}
\end{align}
The remainder 
\begin{align}
\mathcal{R}_N = \cO_{N,t} -   \phi ( h)
\end{align} 
is the sum of $(2^k-2k)$ terms. The estimates 
\begin{align}
\| a(f) \xi \| \leq \| f \|_2 \| \cN^{1/2} \xi \|, \quad \| a^*(f) \xi \| \leq \|f \|_2 \| (\mathcal{N}+1)^{1/2} \xi \|
\end{align}
for any $f \in L^2( \bR^3)$ and any Fock space vector $\xi \in \cF$ together with  \eqref{eq:CLT2-id1} yield the upper bound 
\begin{align}
\label{eq:estimate-R}
\| \mathcal{R}_{N} \xi \| \leq C_k \| O^{(k)} \|_{\rm op } \sum_{j=1}^{2k} N^{-j/2} \| \left( \mathcal{N} + 1\right)^{(j+1)/2} \xi \| 
\end{align}
for any $\xi \in \cF$.  We use the fundamental theorem of calculus to write 
\begin{align}
 \bE_{\xi_{N,t}}  & \left[ e^{ i N^{-k+1/2} O_{N,t}^{(k)}} \right] - \bE_{\cU_\infty (t;0) \Omega} \left[ e^{i \phi (h_t)} \right] \notag\\
=& \expval{ e^{i \tau  N^{-k+1/2} \cO_{N,t} }  } {\cU_\infty (t;0) \Omega}  - \expval{e^{i \tau \phi (h_{t})}}{\cU_\infty (t;0) \Omega} \notag \\
&=  - \int_0^1 ds \; \frac{d}{ds} \expval{ e^{i (1 -s) \cO_{N,t}} e^{is \phi(h_t)}}{\cU_\infty (t;0) \Omega} \notag\\
&= \int_0^1 ds \expval{ e^{i (1 -s) \cO_{N,t}} \mathcal{R}_N e^{is \phi(h_t)}}{\cU_\infty (t;0) \Omega}
\end{align}
Then, it follows from \eqref{eq:estimate-R} 
\begin{align}
\big\vert  \bE_{\xi_{N,t}}  & \left[ e^{ i N^{-k+1/2} O_{N,t}^{(k)}} \right] - \bE_{\cU_\infty (t;0) \Omega} \left[ e^{i \phi (h_t)} \right] \notag\\ 
& \leq  \int_0^1  ds \; \| \mathcal{R}_N e^{is \phi (h_t)} \cU_\infty (t;0) \Omega\| \notag\\
&\leq C_k \| O^{(k)}\|_{\rm op} \int_0^1 ds \sum_{j=1}^{2k-1} N^{-j/2} \|  \left( \mathcal{N} + 1 \right)^{(j+1)/2} e^{is \phi (h_t)} \cU_\infty (t;0) \Omega\| \label{eq:CLT2-3}
\end{align}
With \cite[Proposition 3.4]{BSS} and $\| h_t \|_2^2 \leq  \| O^{(k)} \|_{\rm op}^2$ by definition \eqref{eq:def-h}, we have 
\begin{align}
 \| & \left( \mathcal{N} + 1 \right)^{(j+1)/2} e^{is \phi (h_t)} \cU_\infty (t;0) \Omega\|  \leq C \| \left( \mathcal{N} + 1 + s^2 \| O^{(k)} \|_{\rm op}^2 \right)^{(j+1)/2}\cU_\infty (t;0) \Omega\|  
\end{align}
and furthermore, with \cite[Lemma 3.2]{BSS}
\begin{align}
 \| & \left( \mathcal{N} + 1 \right)^{(j+1)/2} e^{is \phi (h_t)} \cU_\infty (t;0) \Omega\|  \notag\\
& \hspace{1cm}\leq C e^{C \vert t \vert } \| \left( \mathcal{N} + 1 + s^2 \| O^{(k)} \|_{\rm op}^2 \right)^{(j+1)/2} \Omega\| \leq  C e^{C \vert t \vert } \big(1 + s^2 \| O^{(k)} \|_{\rm op}^2 \big)^{(j+1)/2} \label{eq:CLT2-4}
\end{align}
We use now the estimate \eqref{eq:CLT2-4} for \eqref{eq:CLT2-3} and arrive at 
\begin{align}
\big\vert  \bE_{\xi_{N,t}}  & \left[ e^{ i N^{-k+1/2} O_{N,t}^{(k)}} \right] - \bE_{\cU_\infty (t;0) \Omega} \left[ e^{i \phi (h_t)} \right] \notag\\  
& \leq C_k \; e^{C \vert t \vert } \;  \| O^{(k)}\|_{\rm op}  \sum_{j=1}^{2k-1} N^{-j/2}\int_0^1 ds \;  \left(1 + s^2 \| O^{(k)} \|_{\rm op}^2 \right)^{(j+1)/2} \notag \\
& \leq C_k \; e^{C \vert t \vert }  \| O^{(k)}\|_{\rm op} \sum_{j=1}^{2k-1} N^{-j/2} \left( 1 +  \| O^{(k)} \|_{\rm op}^2 \right)^{(j+1)/2} ,
\end{align}
which proves the lemma. 
\end{proof} 

\begin{lemma}
\label{lemma:step3}
Under the same assumptions as in Theorem \ref{thm:CLT}, let $f_{t;s} = \sum_{i=1}^k f_{t;s}^{(i)}  \in L^2( \bR^3) $ be given by \eqref{def:f}. Then,  we have 
\begin{align}
\bE_{\cU_\infty (t;0) \Omega} \left[ e^{i \phi (h_t)} \right] = e^{- \| f_{t;0} \|_2^2/2} \; . 
\end{align}
\end{lemma}

\begin{proof} We need to compute the expectation value 
\begin{align}
\bE_{\cU_\infty (t;0) \Omega} \left[ e^{i \phi (h_t)} \right] = \expval{e^{i  \phi (h_t)}}{\cU_\infty (t;0) \Omega}  
\end{align}
We recall that the limiting dynamics $\cU_\infty (t;0)$ defined in \eqref{eq:Uinfty} acts as a Bogoliubov transform. In particular in follows from \eqref{eq:bogo}, \eqref{eq:bogodyn} and the notations introduced therein that
\begin{align}
\cU_\infty^* (t;0) \phi( h_t) \cU_\infty (t;0) = \phi( \left[ U(t;0) + JV(t;0) \right] h_t ) =  \phi( f_{0;t} ) \; . 
\end{align}
with $f_{0;t}$ defined in \eqref{def:f}.  Hence, we have 
\begin{align}
\bE_{\cU_\infty (t;0) \Omega} \left[ e^{i \phi (h_t)} \right] =  \expval{e^{i  \cU_\infty (t;0)^* \phi (h_t)\cU_\infty (t;0) }}{ \Omega}  = \expval{e^{i  \phi (f_{0;t})}}{\Omega} 
\end{align}
With the Baker-Campbell-Hausdorff formulas, we  split sum in the exponential and arrive at 
\begin{align}
\bE_{\cU_\infty (t;0) \Omega} \left[ e^{i \phi (h_t)} \right]   =& e^{- \| f_{0;t} \|_2^2 /2} \expval{e^{i a^*(f_{0;t})} e^{i 
 a(f_{0;t}) }}{\Omega }  = e^{- \| f_{0;t} \|_2^2/2 } \; . 
\end{align}
\end{proof}

\begin{proof}[Proof of Theorem \ref{thm:char}]
Combining now Lemma \ref{lemma:step1}, Lemma \ref{lemma:step2} and Lemma \ref{lemma:step3}, we arrive at Theorem \ref{thm:char}. 
\end{proof}

\subsection*{Acknowledgements}  S.R. would like to thank Robert Seiringer and Benedikt Stufler for helpful discussions. Funding from the European Union's Horizon 2020 research and innovation programme under the ERC grant agreement No. 694227 and under the Marie Sk\l{}odowska-Curie  Grant Agreement No. 754411 is gratefully acknowledged.

\end{document}